\documentclass[11pt]{article}
\usepackage{fullpage}

\usepackage{times}
\usepackage{comment,amsfonts,amssymb,amsmath,amsthm,graphicx}
\newcommand{\commentout}[1]{}

\ifx\pdftexversion\undefined
\usepackage[colorlinks,linkcolor=black,filecolor=black,citecolor=black,urlco
lor=black,pdfstartview=FitH]{hyperref}
\else
\usepackage[colorlinks,linkcolor=blue,filecolor=blue,citecolor=blue,urlcolor
=blue,pdfstartview=FitH]{hyperref}
\fi

\newcommand{\alert}[1]{\textbf{\color{red}
[[[#1]]]}\marginpar{\textbf{\color{red}**}}\typeout{ALERT:
\the\inputlineno: #1}}

\newcommand{\R}{\mathbb{R}}

\newcommand{\mommit}[1]{}
\newcommand{\namedref}[2]{\hyperref[#2]{#1~\ref*{#2}}}
\newcommand{\sectionref}[1]{\namedref{Section}{#1}}

\newcommand{\theoremref}[1]{\namedref{Theorem}{#1}}

\newcommand{\figureref}[1]{\namedref{Figure}{#1}}

\newcommand{\lemmaref}[1]{\namedref{Lemma}{#1}}

\newcommand{\questionref}[1]{\namedref{Question}{#1}}

\newtheorem{theorem}{Theorem}
\newtheorem{lemma}{Lemma}

\newtheorem{remark}{Remark}

\newtheorem{conjecture}{Conjecture}
\newtheorem{question}[conjecture]{Question}

\usepackage{pdfsync}

\begin{document}

\title{On the Impossibility of Dimension Reduction for Doubling Subsets of $\ell_p$, $p>2$}

\author{Yair Bartal\thanks{Hebrew University. Email: \texttt{yair@cs.huji.ac.il}.
    Supported in part by a grant from the Israeli Science Foundation (1609/11).
    }
\and
Lee-Ad Gottlieb\thanks{Ariel University. Email: \texttt{leead@ariel.ac.il}.}
\and
Ofer Neiman\thanks{Ben-Gurion University of the Negev. Email:
\texttt{neimano@cs.bgu.ac.il}. Supported in part by ISF grant No. (523/12) and by the European Union's Seventh Framework
Programme (FP7/2007-2013) under grant agreement $n^\circ 303809$.
}
}

\maketitle

\begin{abstract}
A major open problem in the field of metric embedding is the existence of dimension reduction for $n$-point subsets of Euclidean space, such that both distortion and dimension depend only on the {\em doubling constant} of the pointset, and not on its cardinality. In this paper, we negate this possibility for $\ell_p$ spaces with $p>2$. In particular, we introduce an $n$-point subset of $\ell_p$ with doubling constant $O(1)$, and demonstrate that any embedding of the set
into $\ell_p^d$ with distortion $D$ must have $D\ge\Omega\left(\left(\frac{c\log n}{d}\right)^{\frac{1}{2}-\frac{1}{p}}\right)$.
\end{abstract}

\section{Introduction}

Dimension reduction is one of the fundamental tools in algorithms design and a host of
related fields. A particularly celebrated result in this area is the Johnson-Lindenstrauss Lemma
\cite{jl}, which demonstrates that any $n$-point subset of $\ell_2$ can be embedded with arbitrarily
small distortion $1+\epsilon$ into $\ell_2^d$ with $d=O(\log n/\epsilon^2)$. The JL-Lemma has
found applications in such varied fields as machine learning \cite{AV-99,BBV-06}, compressive sensing \cite{BDDW08},
nearest-neighbor (NN) search \cite{KOR98}, information retrieval \cite{V04} and many more.

A limitation of the JL-Lemma is that it is quite specific to $\ell_2$, and in fact there are lower bounds that
rule out dimension reduction for the spaces $\ell_1$ and $\ell_\infty$. Yet essentially no non-trivial bounds are
known for $\ell_p$ when $p\notin\{1,2,\infty\}$. The prospect of dimension reduction for $\ell_p$ would imply
efficient algorithms for NN search and related proximity problems such as clustering, distance oracles and
spanners.

The doubling constant of a metric space $(X,d)$ is the minimal $\lambda$ such that any ball of radius $2r$ can be
covered by $\lambda$ balls of radius $r$, and the doubling dimension of $(X,d)$ is defined as $\log_2\lambda$.
A family of metrics is called {\em doubling} if the doubling constant of each of its members is bounded by some
constant. The doubling dimension is a measure of the {\em intrinsic dimensionality} of a point set.
In the past decade, it has been used in the development and analysis of
algorithms for fundamental problems such as nearest neighbor search \cite{KL04,BKL06,CG06} and
clustering \cite{ABS08,FM10}, for graph problems such as spanner construction
\cite{GGN06,CG06b,DPP06,GR08}, the traveling salesman problem \cite{Talwar04,BGK12},
and routing \cite{KSW04,Slivkins05,AGGM06,KRXY07}, and in machine learning \cite{Bshouty2009323,GKK10}.
Importantly, it has also been observed that the doubling dimension often bounds the quality of embeddings
for a point set, in terms of distortion and dimension
\cite{A83,GKL03,ABN08,CGT08,BRS11,GK11}.

It is known that the dimension bounds of the Johnson-Lindenstrauss Lemma are close to optimal. A simple volume
argument suggests that the set of $n$ standard unit vectors in $\R^n$ requires dimension at least
$\Omega(\log_Dn)$ to embed into in any Euclidean embedding with distortion $D$. Alon \cite{A03} extended this
lower bound to the low distortion regime and demonstrated that any embedding with $1+\epsilon$ distortion
requires $\Omega(\log n/(\epsilon^2\log(1/\epsilon)))$ dimensions, thus showing that the Johnson-Lindenstrauss
Lemma is nearly tight. However, this set of vectors has very high intrinsic dimension, as the doubling
constant is $n$. Hence, it is only natural to ask the following:

\begin{question}\label{question:d}
Do subsets of $\ell_2$ with constant doubling dimension embed into constant dimensional space with
low distortion?
\end{question}

This open question was first raised by \cite{LP01,GKL03}, and is considered among the most important
and challenging problems in the study of doubling spaces \cite{ABN08,CGT08,GK11,NN12,N13}.\footnote{
It follows from simple volume arguments that the best quantitative result that can be hoped for is an embedding with $1+\epsilon$ distortion using $O_\epsilon(\log\lambda)$ dimensions.}
In this work, we consider
the natural counterpart of Question \ref{question:d} for $\ell_p$ spaces ($p>2$), and resolve our
question in the negative:

\begin{theorem}\label{thm:main}
For any $p>2$ there is a constant $c=c(p)$ such that for any positive integer $n$, there is a subset
$A\subseteq\ell_p$ of cardinality $n$ with doubling constant $O(1)$, such that any embedding of $A$ into
$\ell_p^d$ with distortion at most $D$ satisfies
\[
D\ge\Omega\left(\left(\frac{c\log n}{d}\right)^{\frac{1}{2}-\frac{1}{p}}\right)~.
\]
\end{theorem}
This result is the first non-trivial lower bound on dimension reduction in $\ell_p$ spaces where $p\notin
\{1,2,\infty\}$. Additionally, it rules out a class of efficient algorithms for NN-search and the problems
discussed above.

\paragraph{Techniques.}
While there exist numerous techniques for obtaining dimension reduction lower bounds in $\ell_1$
(see related work below), these all seem to be very specific to $\ell_1$ and fail for $p\neq 1,\infty$.
Instead, we present a combinatorical proof for \theoremref{thm:main}, which utilizes a new method
based on potential functions.

The subset $A$ is based on a recursive graph construction which is very popular for obtaining
distortion-dimension tradeoffs \cite{L01,NR02,BC03,GKL03,LN04,LMN05,LM10,ACNN11,LS11,R12}. In these
constructions one starts with a small basic graph, then in each iteration replaces every edge with the
basic graph. In most constructions, the basic graph is very simple (e.g. a $4$-cycle induces the so
called diamond graph) and is often a series-parallel graph. This is very useful for $\ell_1$, as
\cite{GNRS99,CJLV08} showed these graphs embed to $\ell_1$ with constant distortion. However, these
recursive graphs often require distortion $\Omega(\log^{1/p}n)$ for embedding into $\ell_p$ with $p>2$ \cite{GKL03}, where $n$ is the number of vertices, so one cannot use them directly. The novelty in this work is that the instance we produce is not a
graph, but a certain subset of $\ell_p$ which is inspired by the Laakso graph \cite{L01}, the basic
graph for which is depicted in \figureref{fig:L}.

\begin{figure*}
\begin{center}
\includegraphics[scale=0.7]{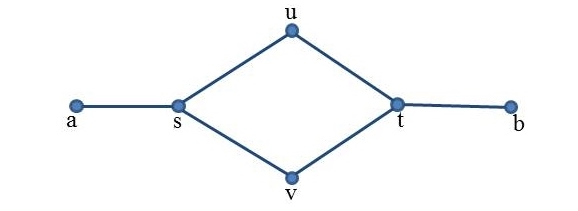}

\end{center}
\caption{The basic instance for the Laakso graph}\label{fig:L}
\end{figure*}

\subsection{Related Work}

Lafforgue and Naor \cite{LN13} have concurrently proved the same result as in \theoremref{thm:main} using analytic tools, with a construction based on the Heisenberg group.

There are several results on embedding metric spaces with low intrinsic dimension into low dimensional normed space with low distortion: Assouad \cite{A83} showed that the snowflakes of doubling metrics\footnote{For $0\le \alpha<1$, an $\alpha$-snowflake of a metric $(X,d)$ is the metric $(X,d^{\alpha})$, that is, all distances are taken to power $\alpha$.} embed with constant distortion into constant dimensional Euclidean space. In particular, this is a positive answer to \questionref{question:d} for this special case.

It is well known that arbitrary $n$ point metrics may require distortion $\Omega(\log n)$ for any embedding into Euclidean space \cite{llr}. It was shown by \cite{GKL03} that any doubling metric embeds with distortion only $O(\sqrt{\log n})$, and that this is best possible. Their result was generalized by \cite{KLMN04} to distortion $O(\sqrt{\log n\cdot\log\lambda})$, which was shown to be tight by \cite{JLM09}. As for low dimensional embedding, \cite{ABN08} showed that any doubling metric may be embedded with distortion $O(\log^{1+\theta}n)$ (for any fixed $\theta>0$) into Euclidean space of dimension proportional to $\log\lambda$, its intrinsic dimension. A trade-off between distortion and dimension for embedding doubling metrics was shown by \cite{CGT08}. For doubling subsets of $\ell_2$, \cite{GK11,BRS11} showed an embedding into constant dimensional $\ell_2$ with $1+\epsilon$ distortion for a {\em snowflake} of the subset. For $\alpha$-snowflakes of arbitrary doubling metrics, \cite{NN12} showed an embedding to Euclidean space where the dimension is a constant independent of $\alpha$ (while the distortion, due to a lower bound of \cite{S96,GKL03,LMN05}, must depend on it).

\paragraph{Lower bounds on dimension reduction} The first impossibility result on dimension reduction in $\ell_1$ is due to \cite{BC03}, who showed that there exists an $n$-point subset of $\ell_1$ that requires $d\ge n^{\Omega(1/D^2)}$ for any $D$-distortion embedding to $\ell_1^d$. Following their work, there have been many different proofs and extensions of this result, using various techniques. The original \cite{BC03} argument was based on linear programming and duality, then \cite{LN04} gave a geometric proof. For the $1+\epsilon$ distortion regime, \cite{ACNN11} used combinatorial techniques to show that the dimension must be at least $n^{1-O(1/\log(1/\epsilon))}$ (and also gave a different proof of the original result). Recently \cite{R12} applied an information theoretic argument to reprove the results of \cite{BC03,ACNN11}. As for {\em linear} dimension reductions, \cite{LMN05} showed a strong lower bound for $\ell_p$ with $p\neq 2$.

The instances used by the papers mentioned above are based on recursive graph constructions. The papers of \cite{BC03,LN04} used the diamond graph, which has high doubling constant, but \cite{LMN05} showed that their proof can be extended to the Laakso graph, yielding essentially the same result but for a subset of low doubling dimension.
For the $\ell_\infty$ space, there are also strong lower bounds,\footnote{For instance, the metric induced by an $n$-point expander graph (which is in $\ell_\infty$ as any other finite metric), requires dimension at least $n^{1/O(D)}$ in any $D$ distortion embedding.} which are based on large girth graphs \cite{matbook} measure concentration \cite{R03} and geometric arguments \cite{LMN05}.

There are few positive results for $p\neq 2$, such as \cite{KOR98} who showed that $\ell_1$ admits dimension reduction when the aspect ratio of the point set is bounded, and
\cite{BG13} used $p$-stable distributions to obtain similar results for all $1<p<2$, and the Mazur map to
obtain a (relatively high-distortion bound) for $p>2$. A weak form of dimensionality reduction in $\ell_1$ was shown by \cite{I06}.

\section{Construction}

Our construction is based on the Laakso graph \cite{L01}, but will lie in $\ell_p$ space. Abusing notation, we will refer to a pair of points as an edge, where some edges will have a level. Fix a parameter $0<\epsilon<1/8$, and for a positive integer $k$ we shall define recursively an instance $A_k=A_k(\epsilon)\subseteq\R^k$. Let $e_0,e_1,\dots,e_k$ be the standard orthonormal basis. In each instance $A_i$ certain pairs of points will be the {\em level $i$ edges}. The initial instance $A_0$ consists of the two points $e_0$ and $-e_0$, which are a level $0$ edge. $A_i$ is created from $A_{i-1}$ by adding to every level $i-1$ edge $\{a,b\}$, four new points $s,t,u,v$ as follows:
\begin{gather*}
s=\frac{3a}{4}+\frac{b}{4}\\
t=\frac{a}{4}+\frac{3b}{4}\\
u=\frac{a}{2}+\frac{b}{2}+ \epsilon\|a-b\|_p\cdot e_i\\
v=\frac{a}{2}+\frac{b}{2}- \epsilon\|a-b\|_p\cdot e_i~,
\end{gather*}
and we will have the following six level $i$ edges: $\{a,s\},\{s,u\},\{s,v\},\{u,t\},\{v,t\},\{t,b\}$. These edges are the {\em child edges} of $\{a,b\}$ (see \figureref{fig:L}). We will refer to the pair $\{u,v\}$ as a {\em diagonal}.

\section{Distortion-Dimension Tradeoff}

Fix any positive integers $d,D,k$, a real $p>2$, and let $\epsilon=\epsilon(d,D,p)$ be the parameter by which we construct the instance $A_k$. The precise value of $\epsilon$ will be determined later. Assume that there is a non-expansive embedding $f$ of $A_k$ into $\ell_p^d$ with distortion $D$, where for each $j\in[d]$ there is a map $f_j:A_k\to\R$ and $f=\bigoplus_{j=1}^df_j$. We want to show a tradeoff between the distortion $D$ and the dimension $d$. The argument is based on the tension between the edges and diagonals - the diagonals tend to contract and the edges to expand. To make this intuition precise, we employ the following potential function for each edge $\{a,b\}$,
\begin{equation}\label{eq:pot}
\Phi(a,b)=\frac{\|f(a)-f(b)\|_2^2}{\|a-b\|_p^2}~.
\end{equation}
Since the embedding is non-expansive, and the image has only $d$ coordinates
\[
\|a-b\|_p\ge\|f(a)-f(b)\|_p\ge\|f(a)-f(b)\|_2\cdot d^{1/p-1/2}~,
\]
raising to power $2$ and rearranging we obtain that the potential of any edge $\{a,b\}$ is never larger than
\begin{equation}\label{eq:pup}
\Phi(a,b)\le d^{1-2/p}~.
\end{equation}
The main goal, which is captured in the following Lemma, is to show that for a suitable choice of $\epsilon$, the potential increases (additively) by some positive number $\alpha=(\epsilon/D)^2$ at every level. Using \eqref{eq:pup} it must be that
\begin{equation}\label{eq:11}
k\le d^{1-2/p}/\alpha~,
\end{equation}
as otherwise the potential of some level $k$ edge will be at least $\alpha k>\alpha\cdot d^{1-2/p}/\alpha =d^{1-2/p}$. The intuition behind the proof of the Lemma is simple: if $\{a,b\}$ is a level $i-1$ edge with potential value $\phi$, then consider the diagonal $\{u,v\}$ created from it in level $i$. Since $u,v$ have the same distance to $a,b$, the potential will be minimized\footnote{More accurately, the maximum potential of a child edge of $\{a,b\}$ will be minimized..}  if $u,v$ are embedded into the same point in space. But in order to provide sufficient contribution for the diagonal $\{u,v\}$ we must have $u,v$ spaced out, which then causes some edges to expand, and thus increases the potential. The technical part of the proof balances between the loss in the potential (incurred because our instance lies in $\ell_p$ space), and the gain to the potential arising from the fact that $\|f(u)-f(v)\|_p$ must be large enough.
\begin{lemma}\label{lem:main}
There exists a constant $c=c(p)$ depending only on $p$ such that when $\epsilon\le d^{-1/p}\cdot D^{-2/(p-2)}/c$ the following holds. For any level $i-1$ edge $\{a,b\}$ with potential value $\phi=\Phi(a,b)$, there exists a level $i$ edge with potential value at least $\phi+(\epsilon/D)^2$.
\end{lemma}

\begin{proof}
Let $s,t,u,v$ be the four new points introduced in $A_i$ from the edge $\{a,b\}$ as described above. For each $j\in [d]$ and $u,v$ we shall define $\Delta_j(u),\Delta_j(v)\in\R$ to be the values such that
\begin{gather}
f_j(u)=\frac{f_j(a)}{2}+\frac{f_j(b)}{2}+\Delta_j(u)\cdot\|a-b\|_p\label{eq:fu}\\
f_j(v)=\frac{f_j(a)}{2}+\frac{f_j(b)}{2}+\Delta_j(v)\cdot\|a-b\|_p\label{eq:fv}~.
\end{gather}
In what follows we show that the sum of the squares of the values $\Delta_j(u),\Delta_j(v)$ must be large, because the diagonal $\{u,v\}$ has sufficient contribution.
Since the embedding has distortion $D$ we have that $\|u-v\|_p/D\le\|f(u)-f(v)\|_p$. As $\|u-v\|_p=2\epsilon\|a-b\|_p$ it must be that
\begin{eqnarray*}
2\epsilon/D\cdot\|a-b\|_p&\le&\|f(u)-f(v)\|_p\\
&=& \|a-b\|_p\left(\sum_{j=1}^d|\Delta_j(u)-\Delta_j(v)|^p\right)^{1/p}\\
&\le& \|a-b\|_p\left(\sum_{j=1}^d2^{p-1}|\Delta_j(u)|^p+\sum_{j=1}^d2^{p-1}|\Delta_j(v)|^p\right)^{1/p}~.
\end{eqnarray*}
Hence for at least one of $u,v$, say w.l.o.g for $u$, it follows that
\[
\left(\sum_{j=1}^d|\Delta_j(u)|^p\right)^{1/p}\ge\epsilon/D~.
\]
Using that the $\ell_2$ norm is larger than the $\ell_p$ norm for $p>2$ we get that
\begin{equation}\label{eq:sum}
\sum_{j=1}^d\Delta_j(u)^2\ge(\epsilon/D)^2~.
\end{equation}
Next we consider the following two quantities:
\begin{gather*}
\Phi'(u,a)=\frac{4\|f(u)-f(a)\|_2^2}{\|a-b\|_p^2}~,\\
\Phi'(u,b)=\frac{4\|f(u)-f(b)\|_2^2}{\|a-b\|_p^2}~.
\end{gather*}
Note that by \eqref{eq:fu}
\begin{eqnarray*}
\Phi'(u,a)&=&\frac{4}{\|a-b\|_p^2}\sum_{j=1}^d\left(\frac{f_j(b)-f_j(a)}{2}+\Delta_j(u)\cdot\|a-b\|_p\right)^2\\
&=&\Phi(a,b)+\sum_{j=1}^d4\Delta_j(u)^2+\sum_{j=1}^d\frac{4(f_j(b)-f_j(a))\cdot\Delta_j(u)}{\|a-b\|_p} ~.
\end{eqnarray*}
Similarly using \eqref{eq:fv}
\begin{eqnarray*}
\Phi'(u,b)&=&\frac{4}{\|a-b\|_p^2}\sum_{j=1}^d\left(\frac{f_j(b)-f_j(a)}{2}-\Delta_j(u)\cdot\|a-b\|_p\right)^2\\
&=&\Phi(a,b)+\sum_{j=1}^d4\Delta_j(u)^2-\sum_{j=1}^d\frac{4(f_j(b)-f_j(a))\cdot\Delta_j(u)}{\|a-b\|_p}~.
\end{eqnarray*}
As the two terms only differ by the sign before $\sum_{j=1}^d\frac{4(f_j(b)-f_j(a))\cdot\Delta_j(u)}{\|a-b\|_p}$, we conclude that at least one of them, assume w.l.o.g $\Phi'(u,a)$, must be at least
\begin{equation}\label{eq:low}
\Phi'(u,a)\ge \Phi(a,b)+\sum_{j=1}^d4\Delta_j(u)^2~.
\end{equation}
Now we shall consider
\begin{gather*}
\Phi'(a,s)=\frac{16\|f(a)-f(s)\|_2^2}{\|a-b\|_p^2}~,\\
\Phi'(u,s)=\frac{16\|f(u)-f(s)\|_2^2}{\|a-b\|_p^2}~.
\end{gather*}
Observe that since $\|a-s\|_p=\frac{1}{4}\|a-b\|_p$ we have that
\begin{equation}\label{eq:as}
\Phi'(a,s)=\Phi(a,s)~.
\end{equation}
Since $e_i$ is a unit vector orthogonal to the subspace in which $s$ lie, we have that $\|u-s\|_p=\frac{1}{4}\|a-b\|_p(1+(4\epsilon)^p)^{1/p}$, and then
\begin{equation}\label{eq:us}
\Phi'(u,s)=(1+(4\epsilon)^p)^{1/p}\cdot\Phi(u,s)~.
\end{equation}
Next, note that
\begin{eqnarray*}
\Phi'(u,a)&=&\frac{4\|f(u)-f(s)+f(s)-f(a)\|_2^2}{\|a-b\|_p^2}\\
&\le&\frac{8\left(\|f(a)-f(s)\|_2^2+\|f(s)-f(u)\|_2^2\right)}{\|a-b\|_p^2}\\
&=&\frac{1}{2}(\Phi'(u,s)+\Phi'(a,s))~,
\end{eqnarray*}
so one of $\Phi'(u,s)$, $\Phi'(a,s)$ is at least as large as $\Phi'(u,a)$.
If it is the case that $\Phi'(a,s)\ge\Phi'(u,a)$, then the assertion of the Lemma is proved for the level $i$ edge $\{a,s\}$, because by \eqref{eq:as}, \eqref{eq:sum} and \eqref{eq:low}
\[
\Phi(a,s)=\Phi'(a,s)\ge\Phi'(u,a)\ge \Phi(a,b)+\sum_{j=1}^d4\Delta_j(u)^2\ge \Phi(a,b)+4(\epsilon/D)^2~.
\]
So from now on we focus on the case that $\Phi'(u,s)\ge\Phi'(u,a)$. The same calculation shows that
\[
\Phi'(u,s)\ge\Phi'(u,a)\ge \Phi(a,b)+4(\epsilon/D)^2~,
\]
and by \eqref{eq:us}
\[
\Phi(u,s)\ge\frac{\Phi(a,b)+4(\epsilon/D)^2}{(1+(4\epsilon)^p)^{1/p}}~.
\]
We will use that for $0<x<1/2$, $e^x\le 1+2x$, then as $\epsilon<1/8$,
\[
(1+(4\epsilon)^p)^{1/p}\le e^{(4\epsilon)^p/p}\le 1+2(4\epsilon)^p/p\le 1+(4\epsilon)^p\le\frac{1}{1-(4\epsilon)^p}~.
\]
Using this, that $1-(4\epsilon)^p>1/2$ and by \eqref{eq:pup},
\begin{eqnarray*}
\Phi(u,s)&\ge&\left(\Phi(a,b)+4(\epsilon/D)^2\right)\cdot(1-(4\epsilon)^p)\\
&\ge&\Phi(a,b)-(4\epsilon)^p\cdot d^{1-2/p} + 2(\epsilon/D)^2~.
\end{eqnarray*}
Recall that $\epsilon\le d^{-1/p}\cdot D^{-2/(p-2)}/c$, and we can set $c=4^{p/(p-2)}$ so that
\[
2(\epsilon/D)^2-(4\epsilon)^p\cdot d^{1-2/p}\ge (\epsilon/D)^2~,
\]
which satisfies the assertion of the Lemma for the edge $\{u,s\}$.
\end{proof}

Finally, let us prove the main Theorem.\footnote{Our proof fails for $p=2$ because this is the place where the loss of $(4\epsilon)^p\cdot d^{1-2/p}$ overcomes the gain of $2(\epsilon/D)^2$.}
\begin{proof}[Proof of \theoremref{thm:main}]
The set $A_k$ has a doubling constant $O(1)$ as established in \sectionref{sec:doubling}. The number of points in $A_k$ is $n=\Theta(6^k)$, so $k=\Theta(\log n)$. Using \eqref{eq:11} we have that
\[
\Omega(\log n)\le d^{1-2/p}\cdot (D/\epsilon)^2 = d^{1-2/p}\cdot D^2\cdot d^{2/p}\cdot  D^{4/(p-2)}\cdot c^2=c^2d\cdot D^{2p/(p-2)}~.
\]
So if one fixes the dimension $d$, we conclude that the distortion must be at least
\begin{equation}\label{eq:trade}
D\ge\Omega_p\left(\left(\frac{\log n}{d}\right)^{1/2-1/p}\right)~.
\end{equation}
\end{proof}

\begin{remark}
A similar calculation shows that if the embedding of $A_k$ is done into $\ell_q$ for some $q\ge 1$, then the tradeoff is
\[
D\ge\Omega_p\left(\frac{\log ^{1/2-1/p}n}{d^{\frac{\max\{q-2,2-q\}}{2q}}}\right)~,
\]
(this is better than \eqref{eq:trade} for $p/(p-1)\le q\le p$).
\end{remark}

\section{$A_k$ is Doubling}\label{sec:doubling}

To prove that instance $A_k$ is doubling, we will construct a recursive instance $B_k \supset A_k$ with infinite
number of points, and show that $B_k$ has constant doubling dimension. Since $A_k$ is a subset of $B_k$, it must
have constant doubling dimension as well \cite{GK13}. $B_k$ is constructed as $A_k$, except that for any level $i$ edge $\{a,b\}$ we add all the points on the line segment connecting $a$ to $b$.
The instance $B_k$ consists of these segments. For the sake of analysis we need that a level $i$ segment will also belong to all the next levels, where in level $i+j$ it will be partitioned to $4^j$ equal length segments. Note that a level $i$ segment created from the edge $\{a,b\}$ will have $8$ segments in level $i+1$ - the $6$ level $i+1$ edges and in addition the segment between $s,t$ is partitioned to two level $i+1$ segments. Observe that each point intersects at most $4$ segments.

First consider a level $i$ edge $\{a,b\}$ in $B_i$ and its child edges in $B_{i+1}$. If $\{a,b\}$ has length $r$,
then the child edges $\{a,s\},\{t,b\}$ have length $\frac{r}{4}$, and the child edges
$\{s,u\},\{s,v\},\{u,t\},\{v,t\}$ have length $\frac{r\cdot(1+(4\epsilon)^p)^{1/p}}{4}$. Since all points in the child
segments are within distance $\epsilon r$ of the parent segment, we may conclude that for any $r$-length segment, all its descendant points are within distance
\begin{equation}\label{eq:kkl}
\epsilon r \sum_{i=0}^{\infty} \left(\frac{(1+(4\epsilon)^p)^{1/p}}{4}\right)^i < 2 \epsilon r~,
\end{equation}
(as $\epsilon<1/8$) from the parent segment.

Define the distance between two segments to be the distance between the closest pair of points on the segments. We say that segments are {\em disjoint} if their intersection is empty.
We will need the following lemma.

\begin{lemma}\label{lem:adjacent}
Let $E$ be a collection of disjoint segments in $B_k$ (from all levels) with lengths in the range
$[r,4r]$ and inter-segment distance at most $10 \epsilon r$. Then $|E| = O(1)$.
\end{lemma}

\begin{proof}
We first claim that the distance between any two disjoint segments in $B_k$ of length in $[r,4r]$ is at least
$\frac{\epsilon}{2} r$. To prove this, note that it holds trivially for any pair of disjoint segments in
$B_1$, and so it also holds between descendants of these disjoint segments. It remains to prove
that the claim holds between descendants of adjacent segment pairs in $B_1$. We will consider all adjacent
segment pairs in turn, first focusing on the pair of line segments intersecting at $u$ (equivalently, at $v$)
and then on the four line segments intersecting at $s$ (or $t$).

Let us take the segment pair which intersects at $u$, and call this pair instance $D$: Expanding $D$ to the
children of the edges, it is easily verified that the claim holds between disjoint child segment pairs, and
so it must hold between their descendants as well. We must show that the claim holds between the descendants
of adjacent child segment pairs. Observe that all adjacent segment pairs are covered by three instances: Two
$B$ instances (a scaled copy of $B_1$) cover the respective children of each parent segment, and one $D$ instance covers the child
segment pair which (like their parents) intersect at $u$. Then the proof has successfully progressed a single
level in the construction.

Now let instance $F$ be the four adjacent edges intersecting at $s$. We expand $F$ to its children, and it is
readily verified that the claim holds between pairs of disjoint child segments, and therefore between their
descendants as well. We must show that the claim holds between descendants of adjacent child segment pairs.
As above, the adjacent pairs can all be covered by five instances, four of type $B$ and one of type $F$. It
follows that the proof of the claim holds inductively throughout all levels.

Now consider the collection $E$. Note that the great-grandparents of the segments in $E$ have length at least
$r' = r\cdot(4/(1+(4\epsilon)^p)^{1/p})^3\ge 30r$, and at most $r\cdot 4^3\le 4r'$. By the assertion of the Lemma and \eqref{eq:kkl} the distance between these great-grandparents is at most $10 \epsilon r +4\epsilon r< \frac{\epsilon}{2} r'$, and the claim above shows they must all be non-disjoint. The maximum number of mutually adjacent segments is $4$, so we have at most $4$ great-grandparents to the segments of $E$. Since a single segment begets
eight other segments in each level, $|E| \le 4 \cdot 8^3$ segments.
\end{proof}

We can now bound the doubling dimension of the point set. Suppose there exists a set $S \in B_k$ of points
with inter-point distances in the range $[\gamma,2\gamma]$, then the doubling constant is at most $|S|^2$ \cite{GK13}. For each point $x\in S$, we find its closest ancestral segment $s(x)$ of
length at least $\frac{\gamma}{4\epsilon}$ (and at most $\frac{\gamma}{\epsilon}$) and project $x$ onto $s(x)$. By \eqref{eq:kkl} the distance of $x$ from $s(x)$ is at most $\frac{\gamma}{2}$, thus the projected points have inter-point distance in the range
$[\frac{\gamma}{2},\frac{5\gamma}{2}]$. The distance range implies that any segment can have at most $6$ projected points incident upon it. By \lemmaref{lem:adjacent}, the number
of segments of length in $[r,4r]$ where $r = \frac{\gamma}{4\epsilon}$ with inter-segment distance at most
$\frac{5\gamma}{2} = 10 \epsilon r$ is bounded by a constant.

\bibliographystyle{alpha}
\bibliography{art,adi}

\end{document}